\documentclass[11pt,twoside]{article}
\usepackage{amsmath,amssymb,latexsym,graphicx,hyperref,epstopdf,subcaption,color}
\usepackage{algpseudocode,algorithm,epstopdf}
\usepackage[capitalise]{cleveref}
\numberwithin{figure}{section}
\normalsize
\newtheorem{theorem}{Theorem}[section]
\newtheorem{lemma}{Lemma}[section]

\newtheorem{claim}{Claim}[section]

\newtheorem{definition}{Definition}[section]

\newenvironment{proof}{{\sc Proof. }}{\hfill$\Box$\vspace{0.2in}}

\def\mcC{\mathcal{C}}

\def\mcQ{\mathcal{Q}}

\title{A local search $4/3$-approximation algorithm for the minimum $3$-path partition problem}
\author{Yong~Chen\thanks{\texttt{Emails:\{chenyong,anzhang\}@hdu.edu.cn}.
	Department of Mathematics, Hangzhou Dianzi University.  Hangzhou, Zhejiang, China.}
	\and
	Randy~Goebel\thanks{\texttt{Emails:\{rgoebel,guohui,xu2\}@ualberta.ca}.
	Department of Computing Science, University of Alberta. Edmonton, Alberta T6G 2E8, Canada.}
	\and
	Guohui~Lin$^{\dagger}$\thanks{Correspondence authors.}
	\and
	Longcheng~Liu\thanks{\texttt{Email:longchengliu@xmu.edu.cn}.
	School of Mathematical Sciences, Xiamen University. Xiamen, Fujian, China.} $^{\dagger}$
	\and
	Bing~Su\thanks{\texttt{Email:subing684@sohu.com}.
	School of Economics and Management, Xi'an Technological University. Xi'an, Shaanxi, China.}
	\and
	Weitian~Tong\thanks{\texttt{Email:wtong@georgiasouthern.edu}.
	Department of Computer Science, Georgia Southern University. Statesboro, Georgia, USA.} $^{\dagger}$
	\and
	Yao~Xu$^{\dagger}$$^{\ddagger}$
	\and
	An~Zhang$^*$}%
\date{\today}
\begin{document}
\maketitle

\begin{abstract}
Given a graph $G = (V, E)$, the $3$-path partition problem is to find a minimum collection of vertex-disjoint paths
each of order at most $3$ to cover all the vertices of $V$.
It is different from but closely related to the well-known $3$-set cover problem.
The best known approximation algorithm for the $3$-path partition problem was proposed recently and has a ratio $13/9$.
Here we present a local search algorithm 
and show, by an amortized analysis, that it is a $4/3$-approximation.
This ratio matches up to the best approximation ratio for the $3$-set cover problem.

\end{abstract}

\section{Introduction}
Motivated by the data integrity of communication in wireless sensor networks and several other applications,
the {\sc $k$-path partition} ($k$PP) problem was first considered by Yan et al.~\cite{YCH97}.
Given a simple graph $G = (V, E)$ (we consider only simple graphs), with $n = |V|$ and $m = |E|$,
the {\em order} of a simple path in $G$ is the number of vertices on the path and it is called a {\em $k$-path} if its order is $k$.
The $k$PP problem is to find a minimum collection of vertex-disjoint paths each of order at most $k$
such that every vertex is on some path in the collection.

Clearly, the $2$PP problem is exactly the {\sc Maximum Matching} problem, which is solvable in $O(m \sqrt{n} \log(n^2/m)/\log n)$-time~\cite{GK04}.
For each $k \ge 3$, $k$PP is NP-hard~\cite{GJ79}.
We point out the key phrase ``at most $k$'' in the definition,
that ensures the existence of a feasible solution for any given graph;
on the other hand, if one asks for a path partition in which every path has an order exactly $k$,
the problem is called {\em $P_k$-partitioning} and is also NP-complete for any fixed constant $k \ge 3$~\cite{GJ79},
even on bipartite graphs of maximum degree three~\cite{MT07}.
To the best of our knowledge, there is no approximation algorithm with proven performance for the general $k$PP problem,
except the trivial $k$-approximation using all $1$-paths.
For {\sc $3$PP}, Monnot and Toulouse~\cite{MT07} proposed a $3/2$-approximation, based on two maximum matchings;
recently, Chen et al.~\cite{CGL18} presented an improved $13/9$-approximation.

The $k$PP problem is a generalization to the {\sc Path Cover} problem~\cite{FR02} (also called {\sc Path Partition}),
which is to find a minimum collection of vertex-disjoint paths which together cover all the vertices in $G$.
{\sc Path Cover} contains the {\sc Hamiltonian Path} problem~\cite{GJ79} as a special case,
and thus it is NP-hard and it is outside APX unless P = NP.

The $k$PP problem is also closely related to the well-known {\sc Set Cover} problem.
Given a collection of subsets $\mcC = \{S_1, S_2, \ldots, S_m\}$ of a finite ground set $U = \{x_1, x_2, \ldots, x_n\}$,
an element $x_i \in S_j$ is said to be {\em covered} by the subset $S_j$, and
a {\em set cover} is a collection of subsets which together cover all the elements of the ground set $U$.
The {\sc Set Cover} problem asks to find a minimum set cover.
{\sc Set Cover} is one of the first problems proven to be NP-hard~\cite{GJ79},
and is also one of the most studied optimization problems for the approximability~\cite{Joh74} and inapproximability~\cite{RS97,Fei98,Vaz01}.
The variant of {\sc Set Cover} in which every given subset has size at most $k$ is called {\sc $k$-Set Cover},
which is APX-complete and admits a $4/3$-approximation for $k = 3$~\cite{DF97}
and an $(H_k - \frac {196}{390})$-approximation for $k \ge 4$~\cite{Lev06}.

To see the connection between $k$PP and {\sc $k$-Set Cover},
we may take the vertex set $V$ of the given graph as the ground set, and an $\ell$-path with $\ell \le k$ as a subset;
then the $k$PP problem is the same as asking for a minimum {\em exact} set cover.
That is, the $k$PP problem is a special case of the {\em minimum} {\sc Exact Cover} problem~\cite{Kar72},
for which unfortunately there is no approximation result that we may borrow.
Existing approximations for (non-exact) {\sc $k$-Set Cover} do not readily apply to $k$PP,
because in a feasible set cover, an element of the ground set could be covered by multiple subsets.
There is a way to enforce the {\em exactness} requirement in the {\sc Set Cover} problem,
by expanding $\mcC$ to include all the proper subsets of each given subset $S_j \in \mcC$.
But in an instance graph of $k$PP, not every subset of vertices on a path is traceable,
and so such an expanding technique does not apply.
In summary, $k$PP and {\sc $k$-Set Cover} share some similarities, but none contains the other as a special case.

In this paper, we study the {\sc $3$PP} problem.
The authors of the $13/9$-approx-imation~\cite{CGL18}
first presented an $O(nm)$-time algorithm to compute a $k$-path partition with the least $1$-paths, for any $k \ge 3$;
then they applied an $O(n^3)$-time greedy approach to merge three $2$-paths into two $3$-paths whenever possible.
We aim to design better approximations for {\sc $3$PP} with provable performance,
and we achieve a $4/3$-approximation.
Our algorithm starts with a $3$-path partition with the least $1$-paths,
then it applies a local search scheme to repeatedly search for an {\em expected collection} of $2$- and $3$-paths 
and replace it by a strictly smaller {\em replacement collection} of new $2$- and $3$-paths.

The rest of the paper is organized as follows.
In Section~\ref{sec3} we present the local search scheme
searching for all the expected collections of $2$- and $3$-paths.
The performance of the algorithm is proved through an amortized analysis in Section~\ref{sec4}, where we also provide a tight instance.
We conclude the paper in Section~\ref{sec5}.

\section{A local search approximation algorithm} \label{sec3}
The $13/9$-approximation proposed by Chen et al.~\cite{CGL18} applies only one replacement operation,
which is to merge three $2$-paths into two $3$-paths.
In order to design an approximation for {\sc $3$PP} with better performance,
we examine four more replacement operations, each transfers three $2$-paths to two $3$-paths with the aid of a few other $2$- or $3$-paths.
Starting with a $3$-path partition with the least $1$-paths,
our approximation algorithm repeatedly finds a certain expected collection of $2$- and $3$-paths and
replaces it by a replacement collection of one less new $2$- and $3$-paths,
in which the net gain is exactly one.

In Section \ref{sec3.2} we present all the replacement operations to perform on the $3$-path partition with the least $1$-paths.
The complete algorithm, denoted as {\sc Approx}, is summarized in Section \ref{sec3.3}.

\subsection{Local operations and their priorities} \label{sec3.2}
Throughout the local search, the $3$-path partitions are maintained to have the least $1$-paths.
Our four local operations are designed so not to touch the $1$-paths, ensuring that the final $3$-path partition still contains the least $1$-paths.
These operations are associated with different priorities, that is,
one operation applies only when all the other operations of higher priorities (labeled by smaller numbers) fail to apply to the current $3$-path partition.%
\footnote{We remark that the priorities are set up to ease the presentation and the amortized analysis.}
We remind the reader that the local search algorithm is iterative,
and every iteration ends after executing a designed local operation. 
The algorithm terminates when none of the designed local operations applies.

\begin{definition}
\label{def01}
With respect to the current $3$-path partition $\mcQ$, a local {\sc Operation $i_1$-$i_2$-By-$j_1$-$j_2$},
where $j_1 = i_1 - 3$ and $j_2 = i_2 + 2$,
replaces an {\em expected collection} of $i_1$ $2$-paths and $i_2$ $3$-paths of $\mcQ$ 
by a {\em replacement collection} of $j_1$ $2$-paths and $j_2$ $3$-paths on the same subset of $2i_1 + 3i_2$ vertices. 
\end{definition}

In the rest of this section we determine the configurations for all the expected collections and their priorities, respectively.

\subsubsection{{\sc Operation $3$-$0$-By-$0$-$2$}, highest priority $1$}
When three $2$-paths of $\mcQ$ can be connected into a $6$-path in the graph $G$ (see Figure~\ref{fig21} for an illustration),
they form into an expected collection.
By removing the middle edge on the $6$-path, we achieve two $3$-paths on the same six vertices and they form the replacement collection.
In the example illustrated in Figure~\ref{fig21}, 
with the two edges $(u_1, v_2), (u_2, v_3) \in E$ outside of $\mcQ$,
{\sc Operation $3$-$0$-By-$0$-$2$} replaces the three $2$-paths $u_1$-$v_1$, $u_2$-$v_2$, and $u_3$-$v_3$ of $\mcQ$ by two new $3$-paths $v_1$-$u_1$-$v_2$ and $u_2$-$v_3$-$u_3$.

\begin{figure}[ht]
\centering
\includegraphics[width=.18\linewidth]{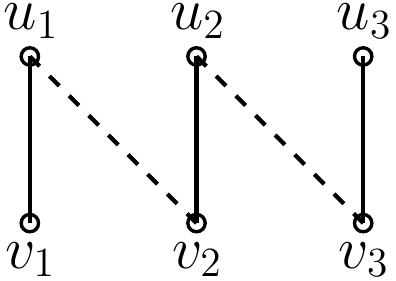}
\caption{The configuration of the expected collection for {\sc Operation $3$-$0$-By-$0$-$2$}, which has the highest priority $1$,
where solid edges are in $\mcQ$ and dashed edges are in $E$ but outside of $\mcQ$.\label{fig21}}
\end{figure}

We point out that {\sc Operation $3$-$0$-By-$0$-$2$} is the only local operation executed in the $13/9$-approximation~\cite{CGL18}.

An {\sc Operation $3$-$0$-By-$0$-$2$} does not need the assistance of any $3$-path of $\mcQ$.
In each of the following operations, we need the aid of one or two $3$-paths of $\mcQ$ to transfer three $2$-paths to two $3$-paths.
We first note that for a $3$-path $u$-$w$-$v \in {\cal Q}$, if $(u, v) \in E$ too,
then if desired, we may {\em rotate} $u$-$w$-$v$ into another $3$-path with $w$ being an endpoint (see Figure~\ref{fig22} for an illustration).
In the following, any $3$-path in an expected collection can be either the exact one in $\mcQ$ or the one rotated from a $3$-path in $\mcQ$.

\begin{figure}[ht]
\centering
\includegraphics[width=.29\linewidth]{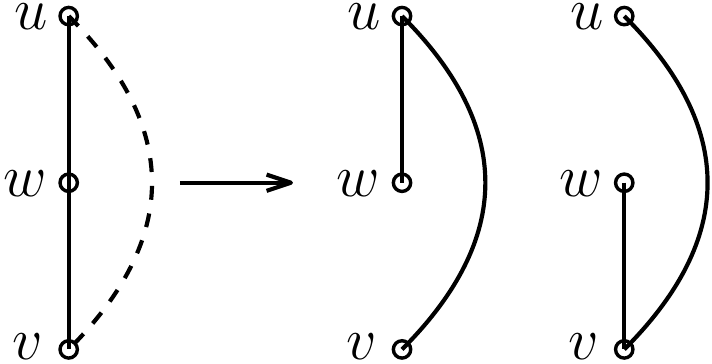}
\caption{A $3$-path $u$-$w$-$v \in \mcQ$ can be rotated so that $w$ becomes an endpoint if $(u, v) \in E$,
where solid edges are in $\mcQ$ and dashed edges are in $E$ but outside of $\mcQ$.\label{fig22}}
\end{figure}

\subsubsection{{\sc Operation $3$-$1$-By-$0$-$3$}, priority $2$}
Consider an expected collection of three $2$-paths $P_1 = u_1$-$v_1$, $P_2 = u_2$-$v_2$, $P_3 = u_3$-$v_3$, and a $3$-path $P_4 = u$-$w$-$v$ in $\mcQ$.
Note that an {\sc Operation $3$-$1$-By-$0$-$3$} applies only when {\sc Operation $3$-$0$-By-$0$-$2$} fails to apply to the current $\mcQ$,
that is, $P_1$, $P_2$, $P_3$ cannot be connected into a $6$-path.%
\footnote{This is one of the places where the priorities ease the presentation, by excluding some cases for discussion.}
We identify only the following two classes of configurations for the expected collection in an {\sc Operation $3$-$1$-By-$0$-$3$}.

In the first class, which has priority $2.1$, 
$u, w, v$ are adjacent to an endpoint of $P_1, P_2, P_3$ in $G$, respectively (see Figure~\ref{fig23} for an illustration).
The operation breaks the $3$-path $u$-$w$-$v$ into three singletons and connects each of them to the respective $2$-path to form the replacement collection of three new $3$-paths.
In the example illustrated in Figure~\ref{fig23},
{\sc Operation $3$-$1$-By-$0$-$3$} replaces the expected collection by three new $3$-paths $u$-$u_1$-$v_1$, $w$-$u_2$-$v_2$, and $v$-$u_3$-$v_3$.

\begin{figure}[ht]
\centering
\includegraphics[width=.2\linewidth]{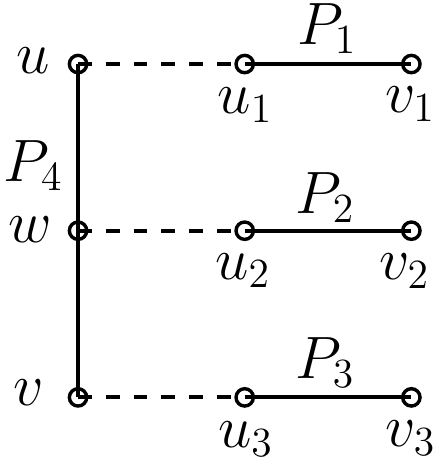}
\caption{The first class of configurations of the expected collection for {\sc Operation $3$-$1$-By-$0$-$3$}, which has priority $2.1$,
where solid edges are in $\mcQ$ and dashed edges are in $E$ but outside of $\mcQ$.\label{fig23}}
\end{figure}

In the second class, which has priority $2.2$,
two of the three $2$-paths, say $P_1$ and $P_2$, are adjacent and thus they can be replaced by a new $3$-path and a singleton.
We distinguish two configurations in this class (see Figure~\ref{fig24} for illustrations).
In the first configuration,
the singleton is adjacent to the midpoint $w$ and $P_3$ is adjacent to one of $u$ and $v$;
in the second configuration,
the singleton and $P_3$ are adjacent to $u$ and $v$, respectively.
For an expected collection of any of the two configurations, the operation replaces it by three new $3$-paths.

In the example illustrated in Figure~\ref{fig24a},
the singleton is $u_1$ and $P_3$ is adjacent to $u$.
{\sc Operation $3$-$1$-By-$0$-$3$} replaces the expected collection by three new $3$-paths $v_1$-$u_2$-$v_2$, $v$-$w$-$u_1$, and $u$-$u_3$-$v_3$.
In the example illustrated in Figure~\ref{fig24b},
the singleton is $u_1$ and $P_3 = u_3$-$v_3$ is adjacent to $u$.
{\sc Operation $3$-$1$-By-$0$-$3$} replaces the expected collection by three new $3$-paths $v_1$-$u_2$-$v_2$, $w$-$v$-$u_1$, and $u$-$u_3$-$v_3$.

\begin{figure}[ht]
\centering
\captionsetup[subfigure]{justification=centering}
\begin{subfigure}{0.4\textwidth}
\centering
\includegraphics[width=.56\linewidth]{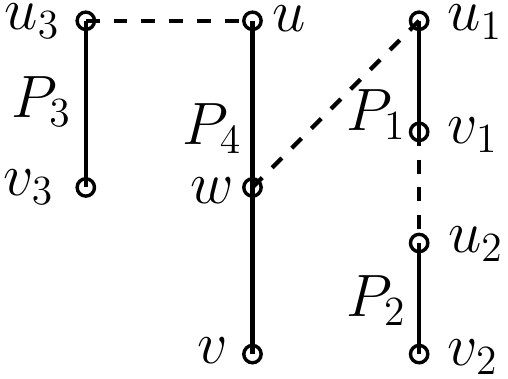}
\caption{\label{fig24a}}
\end{subfigure}
\begin{subfigure}{0.4\textwidth}
\centering
\includegraphics[width=.56\linewidth]{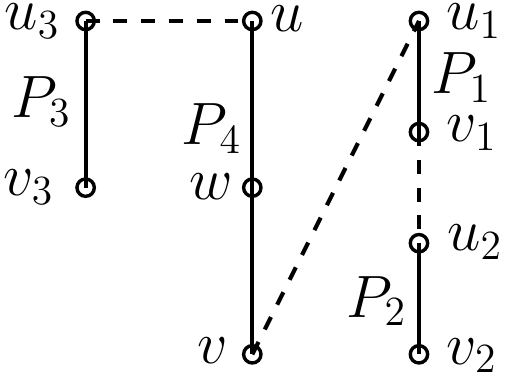}
\caption{\label{fig24b}}
\end{subfigure}
\caption{The second class of configurations of the expected collection in {\sc Operation $3$-$1$-By-$0$-$3$}, which has priority $2.2$,
where solid edges are in $\mcQ$ and dashed edges are in $E$ but outside of $\mcQ$.\label{fig24}}
\end{figure}

\subsubsection{{\sc Operation $4$-$1$-By-$1$-$3$}, priority $3$}
Consider an expected collection of four $2$-paths $P_1 = u_1$-$v_1$, $P_2 = u_2$-$v_2$, $P_3 = u_3$-$v_3$, $P_4 = u_4$-$v_4$,
and a $3$-path $P_5 = u$-$w$-$v$ in $\mcQ$.
Note that an {\sc Operation $4$-$1$-By-$1$-$3$} applies only when {\sc Operation $3$-$0$-By-$0$-$2$} and {\sc Operation $3$-$1$-By-$0$-$3$}
both fail to apply to the current $\mcQ$.%
\footnote{This is another place where the priorities ease the presentation, by excluding some cases for discussion.}
We thus consider only the cases when the four $2$-paths can be separated into two pairs, each of which are adjacent in the graph $G$,
and we can replace them by two new $3$-paths while leaving two singletons which are adjacent to a common vertex on $P_5$.

In the configuration for the expected collection in an {\sc Operation $4$-$1$-By-$1$-$3$}, 
the two singletons are adjacent to a common endpoint, say $u$, of $P_5$ (see Figure~\ref{fig25} for an illustration),
then they can be replaced by a new $2$-path $v$-$w$ and a new $3$-path.
Overall, the operation replaces the expected collection by three new $3$-paths and a new $2$-path.
In the example illustrated in Figure~\ref{fig25},
the two singletons are $u_1$ and $u_3$, and they are both adjacent to $u$.
{\sc Operation $4$-$1$-By-$1$-$3$} replaces the expected collection by three new $3$-paths 
$v_1$-$u_2$-$v_2$, $v_3$-$u_4$-$v_4$, $u_1$-$u$-$u_3$,
and a new $2$-path $w$-$v$.

\begin{figure}[ht]
\centering
\includegraphics[width=.33\linewidth]{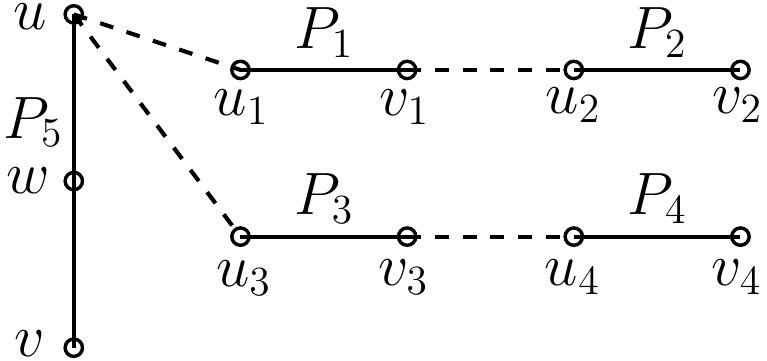}
\caption{The configuration of the expected collection for {\sc Operation $4$-$1$-By-$1$-$3$}, which has priority $3$,
where solid edges are in $\mcQ$ and dashed edges are in $E$ but outside of $\mcQ$.\label{fig25}}
\end{figure}

\subsubsection{{\sc Operation $4$-$2$-By-$1$-$4$}, lowest priority $4$}
Consider an expected collection of four $2$-paths $P_1 = u_1$-$v_1$, $P_2 = u_2$-$v_2$, $P_3 = u_3$-$v_3$, $P_4 = u_4$-$v_4$,
and two $3$-paths $P_5 = u$-$w$-$v$, $P_6 = u'$-$w'$-$v'$ in $\mcQ$.
The four $2$-paths can be separated into two pairs, each of which are adjacent in the graph $G$,
thus we can replace them by two new $3$-paths while leaving two singletons, 
which are adjacent to $P_5$ and $P_6$, respectively (see Figure~\ref{fig26} for illustrations).
We distinguish three classes of configurations for the expected collection in this operation,
for which the replacement collection consists of four new $3$-paths and a new $2$-path.

In the first class, the two singletons are adjacent to $P_5$ and $P_6$ at endpoints, say $u$ and $u'$, respectively;
additionally, one of the five edges
$(u, v')$, $(v, u')$, $(w, v')$, $(v, w')$, $(v, v')$ is in $E$ (see Figure~\ref{fig26a} for an illustration).
In the example illustrated in the Figure~\ref{fig26a}, if $(u, v') \in E$,
then {\sc Operation $4$-$1$-By-$1$-$3$} replaces the expected collection by 
four new $3$-paths $v_1$-$u_2$-$v_2$, $v_3$-$u_4$-$v_4$, $u_1$-$u$-$v'$, $u_3$-$u'$-$w'$,
and a new $2$-path $w$-$v$.

\begin{figure}[ht]
\centering
\captionsetup[subfigure]{justification=centering}
\begin{subfigure}{0.32\textwidth}
\centering
\includegraphics[width=.78\linewidth]{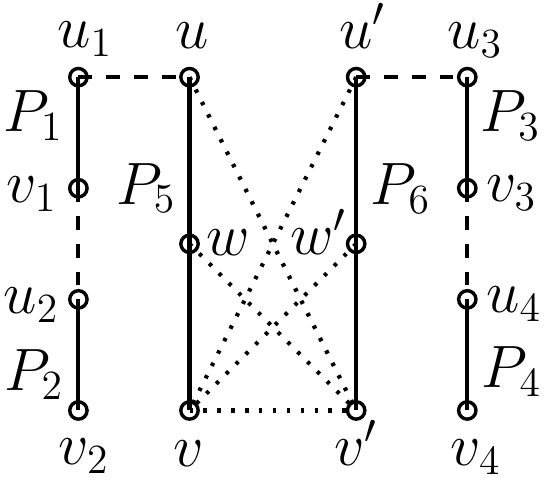}
\caption{The first class.\label{fig26a}}
\end{subfigure}
\hskip .17cm
\begin{subfigure}{0.32\textwidth}
\centering
\includegraphics[width=.78\linewidth]{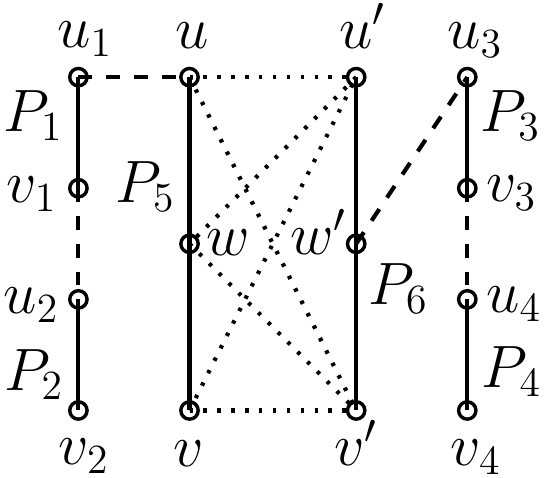}
\caption{The second class.\label{fig26b}}
\end{subfigure}
\hskip .17cm
\begin{subfigure}{0.32\textwidth}
\centering
\includegraphics[width=.78\linewidth]{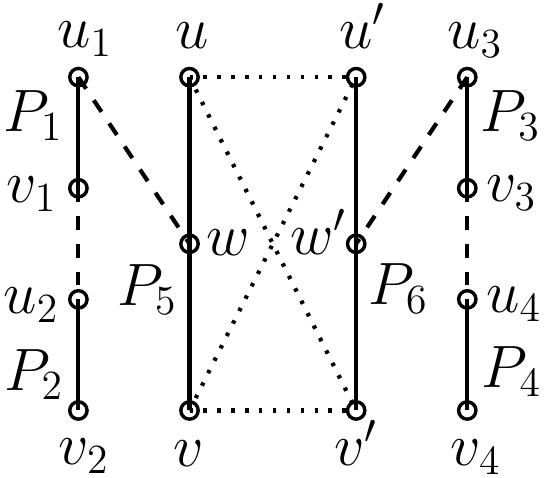}
\caption{The third class.\label{fig26c}}
\end{subfigure}

\caption{The three classes of configurations of the expected collections for an {\sc Operation $4$-$2$-By-$1$-$4$},
	where solid edges are in $\mcQ$, dashed edges are in $E$ but outside of $\mcQ$, and dotted edges could be in $E$ but outside of $\mcQ$.
	In every class, each dotted edge between $P_5$ and $P_6$ corresponds to one configuration.\label{fig26}}
\end{figure}

In the second class, one singleton is adjacent to an endpoint of a $3$-path, say $u$ on $P_5$, and the other singleton is adjacent to the midpoint $w'$ of $P_6$;
additionally, one of the six edges
$(u, u')$, $(u, v')$, $(w, u')$, $(w, v')$, $(v, u')$, $(v, v')$, is in $E$ (see Figure~\ref{fig26b} for an illustration).
In the example illustrated in Figure~\ref{fig26b}, if $(u, u') \in E$,
then {\sc Operation $4$-$1$-By-$1$-$3$} replaces the expected collection by 
four new $3$-paths $v_1$-$u_2$-$v_2$, $v_3$-$u_4$-$v_4$, $u_1$-$u$-$u'$, $u_3$-$w'$-$v'$,
and a new $2$-path $w$-$v$.

In the third class, the two singletons are adjacent to the midpoints of the two $3$-paths, $w$ and $w'$, respectively;
additionally, one of the four edges
$(u, u')$, $(u, v')$, $(v, u')$, $(v, v')$ is in $E$ (see Figure~\ref{fig26c} for an illustration).
In the example illustrated in Figure~\ref{fig26c}, if $(u, u') \in E$,
then {\sc Operation $4$-$1$-By-$1$-$3$} replaces the expected collection by 
four new $3$-paths $v_1$-$u_2$-$v_2$, $v_3$-$u_4$-$v_4$, $u_1$-$w$-$v$, $u_3$-$w'$-$v'$,
and a new $2$-path $u$-$u'$.

\subsection{The complete local search algorithm {\sc Approx}} \label{sec3.3}
The first step of our local search algorithm {\sc Approx} is to compute a $3$-path partition $\mcQ$ with the least $1$-paths.
The second step is iterative,
and in each iteration the algorithm tries to apply one of the four local operations, from the highest priority to the lowest,
by finding a corresponding expected collection and determining the subsequent replacement collection.
When no expected collection can be found, the second step terminates.
The algorithm outputs the achieved $3$-path partition $\mcQ$ as the solution.
A high-level description of the complete algorithm {\sc Approx} is depicted in Figure~\ref{fig27}.

\begin{figure}[H]
\begin{center}
\framebox{
\begin{minipage}{5.5in}
Algorithm {\sc Approx} on $G = (V, E)$:
\begin{description}
\parskip=0pt
\item[Step 1.]
	Compute a $3$-path partition $\mcQ$ with the least $1$-paths in $G$;
\item[Step 2.] Iteratively perform:\\
	2.1. if {\sc Operation $3$-$0$-By-$0$-$2$} applies, update $\mcQ$ and break;\\
	2.2. if {\sc Operation $3$-$1$-By-$0$-$3$} with priority $2.1$ applies, update $\mcQ$ and break;\\
	2.3. if {\sc Operation $3$-$1$-By-$0$-$3$} with priority $2.2$ applies, update $\mcQ$ and break;\\
	2.4. if {\sc Operation $4$-$1$-By-$1$-$3$} applies, update $\mcQ$ and break;\\
	2.5. if {\sc Operation $4$-$2$-By-$1$-$4$} applies, update $\mcQ$ and break;
\item[Step 3.] Return $\mcQ$.
\end{description}
\end{minipage}}
\end{center}
\caption{A high-level description of the local search algorithm {\sc Approx}.\label{fig27}}
\end{figure}

Step 1 runs in $O(nm)$ time~\cite{CGL18}.
Note that there are $O(n)$ $2$-paths and $O(n)$ $3$-paths in $\mcQ$ at the beginning of Step 2,
and therefore there are $O(n^6)$ {\em original} candidate collections to be examined, since a candidate collection has a maximum size of $6$.
When a local operation applies, an iteration ends and the $3$-path partition $\mcQ$ reduces its size by $1$, while introducing at most $5$ new $2$- and $3$-paths.
These new $2$- and $3$-paths give rise to $O(n^5)$ {\em new} candidate collections to be examined in the subsequent iterations.
Since there are at most $n$ iterations in Step 2, we conclude that the total number of original and new candidate collections examined in Step 2 is $O(n^6)$.
Determining whether a candidate collection is an expected collection, and if so, deciding the corresponding replacement collection, can be done in $O(1)$ time.
We thus prove that the overall running time of Step 2 is $O(n^6)$, and consequently prove the following theorem.

\begin{theorem}
\label{thm2.5}
The running time of the algorithm {\sc Approx} is in $O(n^6)$.
\end{theorem}

\section{Analysis of the approximation ratio $4/3$} \label{sec4}
In this section, we show that our local search algorithm {\sc Approx} is a $4/3$-approximation for {\sc $3$PP}.
The performance analysis is done through amortization.

The $3$-path partition produced by the algorithm {\sc Approx} is denoted as $\mcQ$;
let $\mcQ_i$ denote the sub-collection of $i$-paths in $\mcQ$, for $i = 1, 2, 3$, respectively.
Let $\mcQ^*$ be an optimal $3$-path partition, {\it i.e.}, it achieves the minimum total number of paths,
and let $\mcQ_i^*$ denote the sub-collection of $i$-paths in $\mcQ^*$, for $i = 1, 2, 3$, respectively.
Since our $\mcQ$ contains the least $1$-paths among all $3$-path partitions for $G$, we have
\begin{equation}
\label{eq3}
|\mcQ_1| \le |\mcQ_1^*|.
\end{equation}
Since both $\mcQ$ and $\mcQ^*$ cover all the vertices of $V$, we have 
\begin{equation}
\label{eq4}
|\mcQ_1| + 2 |\mcQ_2| + 3 |\mcQ_3| = n = |\mcQ_1^*| + 2 |\mcQ^*_2| + 3 |\mcQ^*_3|.
\end{equation}

Next, we prove the following inequality which gives an upper bound on $|\mcQ_2|$, through an amortized analysis:
\begin{equation}
\label{eq5}
|\mcQ_2| \le |\mcQ^*_1| + 2 |\mcQ^*_2| + |\mcQ^*_3|.
\end{equation}
Combining Eqs.~(\ref{eq3}, \ref{eq4}, \ref{eq5}), it follows that
\begin{equation}
\label{eq6}
3 |\mcQ_1| + 3 |\mcQ_2| + 3 |\mcQ_3| \le 4 |\mcQ_1^*| + 4 |\mcQ^*_2| + 4 |\mcQ^*_3|,
\end{equation}
that is, $|\mcQ| \le \frac 43 |\mcQ^*|$, and consequently the following theorem holds.

\begin{theorem}
\label{thm3.1}
The algorithm {\sc Approx} is an $O(n^6)$-time $4/3$-approximation for the {\sc $3$PP} problem,
and the performance ratio $4/3$ is tight for {\sc Approx}.
\end{theorem}

In the amortized analysis, each $2$-path of $\mcQ_2$ has one token ({\it i.e.}, $|\mcQ_2|$ tokens in total) to be distributed to the paths of $\mcQ^*$.
The upper bound in Eq.~(\ref{eq5}) will immediately follow if we prove the following lemma.

\begin{lemma}
\label{lm3.2}
There is a distribution scheme in which 
\begin{enumerate}
\parskip=2pt
	\item every $1$-path of $\mcQ^*_1$ receives at most $1$ token;
	\item every $2$-path of $\mcQ^*_2$ receives at most $2$ tokens;
	\item every $3$-path of $\mcQ^*_3$ receives at most $1$ token.
\end{enumerate}
\end{lemma}

In the rest of the section we present the distribution scheme that satisfies the three requirements stated in Lemma~\ref{lm3.2}.

Denote $E(\mcQ_2)$, $E(\mcQ_3)$, $E(\mcQ^*_2)$, $E(\mcQ^*_3)$ as the set of all the edges on the paths of $\mcQ_2$, $\mcQ_3$, $\mcQ^*_2$, $\mcQ^*_3$, respectively,
and $E(\mcQ^*) = E(\mcQ^*_2) \cup E(\mcQ^*_3)$.
In the subgraph of $G\big(V, E(\mcQ_2) \cup E(\mcQ^*)\big)$,
only the midpoint of a $3$-path of $\mcQ^*_3$ may have degree $3$, {\it i.e.}, incident with two edges of $E(\mcQ^*)$ and an edge of $E(\mcQ_2)$,
while all the other vertices have degree at most $2$ since each is incident with at most one edge of $E(\mcQ_2)$ and at most one edge of $E(\mcQ^*)$.

Our distribution scheme consists of two phases.
We define two functions $\tau_1(P)$ and $\tau_2(P)$ to denote the fractional amount of token received by a path $P \in \mcQ^*$ in Phase 1 and Phase 2, respectively;
we also define the function $\tau(P) = \tau_1(P) + \tau_2(P)$ to denote the total amount of token received by the path $P \in \mcQ^*$
at the end of our distribution process.
Then, we have $\sum_{P \in \mcQ^*} \tau(P) = |\mcQ_2|$.

\subsection{Distribution process Phase 1} \label{sec4.1}
In Phase 1, we distribute all the $|\mcQ_2|$ tokens to the paths of $\mcQ^*$ ({\it i.e.}, $\sum_{P \in \mcQ^*} \tau_1(P) = |\mcQ_2|$) such that 
a path $P \in \mcQ^*$ receives some token from a $2$-path $u$-$v \in \mcQ_2$ only if $u$ or $v$ is (or both are) on $P$, and
the following three requirements are satisfied:
\begin{enumerate}
\parskip=2pt
	\item $\tau_1(P_i) \le 1$ for $\forall P_i \in \mcQ^*_1$;
	\item $\tau_1(P_j) \le 2$ for $\forall P_j \in \mcQ^*_2$;
	\item $\tau_1(P_\ell) \le 3/2$ for $\forall P_\ell \in \mcQ_3^*$.
\end{enumerate}
In this phase, the one token held by each $2$-path of $\mcQ_2$ is breakable but can only be broken into two halves.
So for every path $P \in \mcQ^*$, $\tau_1(P)$ is a multiple of $1/2$.

For each $2$-path $u$-$v \in \mcQ_2$, at most one of $u$ and $v$ can be a singleton of $\mcQ^*$.
If $P_1 = v \in \mcQ^*_1$, then the whole $1$ token of the path $u$-$v$ is distributed to $v$,
that is, $\tau_1(v) = 1$ (see Figure~\ref{fig31a} for an illustration).
This way, we have $\tau_1(P) \le 1$ for $\forall P \in \mcQ^*_1$.

\begin{figure}[ht]
\centering
\captionsetup[subfigure]{justification=centering}
\begin{subfigure}{0.25\textwidth}
\centering
\includegraphics[width=.38\linewidth]{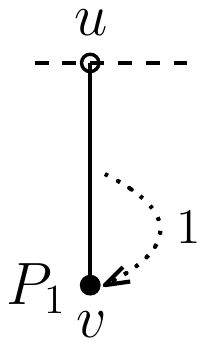}
\caption{\label{fig31a}}
\end{subfigure}
\hskip .15cm
\begin{subfigure}{0.28\textwidth}
\centering
\includegraphics[width=.4\linewidth]{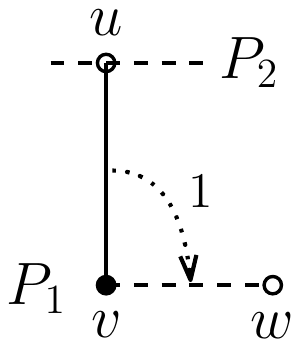}
\caption{\label{fig31b}}
\end{subfigure}
\hskip .15cm
\begin{subfigure}{0.38\textwidth}
\centering
\includegraphics[width=.45\linewidth]{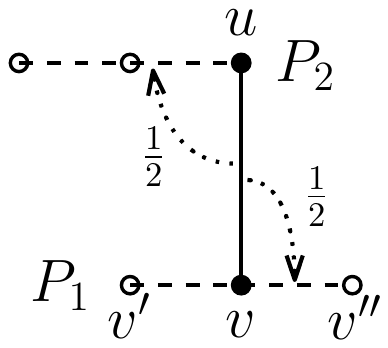}
\caption{\label{fig31c}}
\end{subfigure}

\caption{Illustrations of the token distribution scheme in Phase 1,
where solid edges are in $E(\mcQ_2)$ and dashed edges are in $E(\mcQ^*)$.
In Figure~\ref{fig31c}, $u$ or $v$ can be either an endpoint or the midpoint of the corresponding $3$-path of ${\cal Q}^*_3$.\label{fig31}}
\end{figure}

For a $2$-path $u$-$v \in \mcQ_2$, we consider the cases when both $u$ and $v$ are incident with an edge of $E(\mcQ^*)$.
If one of $u$ and $v$, say $v$, is incident with an edge of $E(\mcQ^*_2)$, that is, $v$ is on a $2$-path $P_1 = v$-$w \in \mcQ^*_2$,
then the $1$ token of the path $u$-$v$ is given to the path $P_1 \in \mcQ^*_2$
(see Figure~\ref{fig31b} for an illustration).
Note that if $u$ is also on a $2$-path $P_2 \in \mcQ^*_2$ and $P_2 \ne P_1$, then the path $P_2$ receives no token from the path $u$-$v$.
The choice of which of the two vertices $u$ and $v$ comes first does not matter.
This way, we have $\tau_1(P) \le 2$ for $\forall P \in \mcQ^*_2$
since the $2$-path $P_1 \in \mcQ^*_2$ might receive another token from a $2$-path of $\mcQ_2$ incident at $w$.

Next, we consider the cases for a $2$-path $u$-$v \in \mcQ_2$ in which each of $u$ and $v$ is incident with an edge of $E(\mcQ^*_3)$.
Consider a $3$-path $P_1 \in \mcQ^*_3$: $v'$-$v$-$v''$.
We distinguish two cases for a vertex of $P_1$ to determine the amount of token received by $P_1$ (see Figure~\ref{fig31c} for an illustration).
In the first case, either the vertex, say $v'$, is not on any path of $\mcQ_2$ or it is on a path of $\mcQ_2$ with $0$ token left,
then $P_1$ receives no token {\em through vertex $v'$}.
In the second case, the vertex, say $v$ (the following argument also applies to the other two vertices $v'$ and $v''$),
is on a path $u$-$v \in \mcQ_2$ holding $1$ token, 
and consequently $u$ must be on a $3$-path $P_2 \in \mcQ^*_3$,
then the $1$ token of $u$-$v$ is broken into two halves, with $1/2$ token distributed to $P_1$ through vertex $v$ and the other $1/2$ token distributed to $P_2$ through vertex $u$.
This way, we have $\tau_1(P) \le 3/2$ for $\forall P \in \mcQ_3^*$
since the $3$-path $P_1 \in \mcQ^*_3$ might receive another $1/2$ token through each of $v'$ and $v''$.

\subsection{Distribution process Phase 2} \label{sec4.2}
In Phase 2, we will transfer the extra $1/2$ token from every $3$-path $P \in \mcQ_3^*$ with $\tau_1(P) = 3/2$ to some other paths of $\mcQ^*$
in order to satisfy the three requirements of Lemma~\ref{lm3.2}.
In this phase, each $1/2$ token can be broken into two quarters,
thus for a path $P \in \mcQ^*$, $\tau_2(P)$ is a multiple of $1/4$.

Consider a $3$-path $P_1 = v''$-$v'$-$v \in \mcQ_3^*$.
We observe that if $\tau_1(P_1) = 3/2$,
then each of $v$, $v'$, and $v''$ must be incident with an edge of $E(\mcQ_2)$, such that the other endpoint of the edge is also on a $3$-path of $\mcQ^*_3$
(see the last case in Phase 1).
Since there are three of them, one of $v$, $v'$ and $v''$, say $v$, on an edge $(u, v) \in E(\mcQ_2)$,
has the associated vertex $u$ on a distinct $3$-path $P_2 \in \mcQ^*_3$ (that is, $P_2 \ne P_1$).
Let $w$ be a vertex adjacent to $u$ on $P_2$, {\it i.e.}, $(u, w)$ is an edge on $P_2$.
(See Figure~\ref{fig32} for an illustration.)
We can verify the following claim.

\begin{claim}
\label{clm3.3}
The vertex $w$ is on a $3$-path of $\mcQ_3$, being either an endpoint or the midpoint.
\end{claim}
\begin{proof}
See Figure~\ref{fig32} for an illustration.
Firstly, $w$ cannot collide into any of $u', u''$
since otherwise the three $2$-paths $u$-$v$, $u'$-$v'$, $u''$-$v''$ could be replaced due to {\sc Operation $3$-$0$-By-$0$-$2$}.
Then, suppose $w$ is on a $2$-path $w$-$x$ of $\mcQ_2$,
then the three $2$-paths $u$-$v$, $u'$-$v'$, $w$-$x$ could be replaced due to {\sc Operation $3$-$0$-By-$0$-$2$}.
Lastly, suppose $w$ is a singleton of $\mcQ_1$,
then $w$ and the $2$-path $u$-$v$ could be merged to a $3$-path so that $\mcQ$ is not a partition with the least $1$-paths, a contradiction.
Thus, $w$ cannot be a singleton of $\mcQ_1$ or on any $2$-path of $\mcQ_2$,
and the claim is proved.
\end{proof}

We conclude from Claim~\ref{clm3.3} that
$\tau_1(P_2) \le 1$,
and we have the following lemma.

\begin{lemma}
\label{lm3.4}
For any $3$-path $P_1 \in \mcQ_3^*$ with $\tau_1(P_1) = 3/2$,
there must be another $3$-path $P_2 \in \mcQ_3^*$ with $\tau_1(P_2) \le 1$ such that
\begin{enumerate}
\parskip=2pt
	\item $u$-$v$ is a $2$-path of $\mcQ_2$, where $v$ is on $P_1$ and $u$ is on $P_2$, and
	\item any vertex adjacent to $u$ on $P_2$ is on a $3$-path $P_3$ of $\mcQ_3$.
\end{enumerate}
\end{lemma}

\begin{figure}[ht]
\centering
\includegraphics[width=0.27\linewidth]{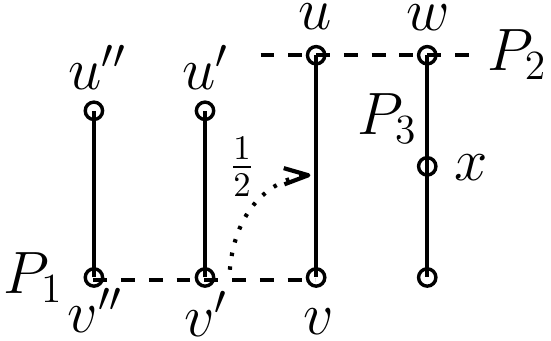}
\caption{An illustration of a $3$-path $P_1 = v$-$v'$-$v'' \in \mcQ_3^*$ with $\tau_1(P_1) = 3/2$, where $u$-$v$, $u'$-$v'$, $u''$-$v'' \in E(\mcQ_2)$.
	The vertex $u$ is on a distinct $3$-path $P_2 \in E(\mcQ^*_3)$;
	the vertex $w$ is on $P_2$ and is adjacent to $u$;
	$w$ is either the midpoint or an endpoint of a $3$-path $P_3 \in \mcQ_3$.\label{fig32}}
\end{figure}

The first step of Phase 2 is to transfer this extra $1/2$ token back from $P_1$ to the $2$-path $u$-$v$ through vertex $v$
(see Figure~\ref{fig32} for an illustration).
Thus, we have $\tau_2(P_1) = -1/2$ and $\tau(P_1) = 3/2 - 1/2 = 1$.

Using Lemma~\ref{lm3.4} and its notation, let $x_1$ and $y_1$ be the other two vertices on $P_3$ ({\it i.e.}, $P_3 = w$-$x_1$-$y_1$ or $P_3 = x_1$-$x$-$y_1$).
Denote $P_4 \in \mcQ^*$ ($P_5 \in \mcQ^*$, respectively) as the path where $x_1$ ($y_1$, respectively) is on.
Next, we will transfer the $1/2$ token from $u$-$v$ to the paths $P_4$ or/and $P_5$ through some {\em pipe} or {\em pipes}.

We define a {\em pipe} $r \to s \to t$,
where $r$ is an endpoint of a {\em source} $2$-path of $\mcQ_2$ ($u$-$v$ here) which receives $1/2$ token in the first step of Phase 2, 
$(r, s)$ is an edge on a $3$-path $P' \in \mcQ^*_3$ with $\tau_1(P') \le 1$ ($P' = P_2$ here),
$s$ and $t$ are both on a $3$-path of $\mcQ_3$ ($P_3$ here),
and $t$ is a vertex on the {\em destination} path of $\mcQ^*$ ($P_4$ or $P_5$ here) which will receive token from the source $2$-path of $\mcQ_2$.
That is, the pipe $r \to s \to t$ will transfer some token from the source $2$-path of $\mcQ_2$ to the destination path of $\mcQ^*$.
$r$ and $t$ are called the {\em head} and {\em tail} of the pipe, respectively.
For example, in Figure~\ref{fig33a}, there are four possible pipes $u \to w \to x_1$, $u \to w \to y_1$, $u'' \to w \to x_1$, and $u'' \to w \to y_1$.
We distinguish the cases,
in which the two destination paths $P_4$ and $P_5$ belong to different combinations of $\mcQ^*_1$, $\mcQ^*_2$, $\mcQ^*_3$,
to determine how they receive token from source $2$-paths through some pipe or pipes.

Recall that $u$ can be either an endpoint or the midpoint of $P_2$.
We distinguish the following cases with $u$ being an endpoint of $P_2$ (the cases for $u$ being the midpoint can be discussed the same), that is, $P_2 = u$-$w$-$u''$, depending on which of $\mcQ^*_1$, $\mcQ^*_2$, $\mcQ^*_3$ the two destination paths $P_4$ and $P_5$ belong to,
to determine the upper bounds on $\tau(P_4)$ and $\tau(P_5)$.

{\bf Case 1.}
At least one of $P_4$ and $P_5$ is a singleton of $\mcQ_1^*$, say $P_4 = x_1 \in \mcQ_1^*$ (see Figure~\ref{fig33} for illustrations).
In this case, we have $\tau_1(P_4) = 0$,
so we transfer the $1/2$ token from $u$-$v$ to $P_4$ through pipe $u \to w \to x_1$.
We observe that if $P_5$ is also a $3$-path of $\mcQ^*_3$, with $(y_1, y_2)$ being an edge on $P_5$,
then $y_2 \to y_1 \to x_1$ is a candidate pipe through which $P_4$ could receive another $1/2$ token.
We distinguish the following two sub-cases based on whether $w$ is an endpoint or the midpoint of $P_3$ to determine all the possible pipes through each of which could $P_4$ receive $1/2$ token.

\begin{figure}[ht]
\centering
\captionsetup[subfigure]{justification=centering}
\begin{subfigure}{0.3\textwidth}
\centering
\includegraphics[width=.85\linewidth]{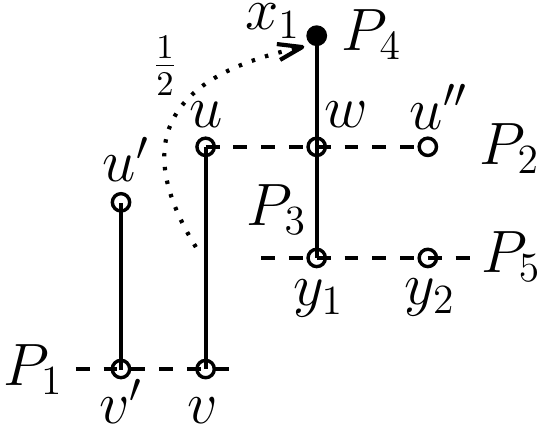}
\caption{\label{fig33a}}
\end{subfigure}
\hskip .15cm
\begin{subfigure}{0.3\textwidth}
\centering
\includegraphics[width=.85\linewidth]{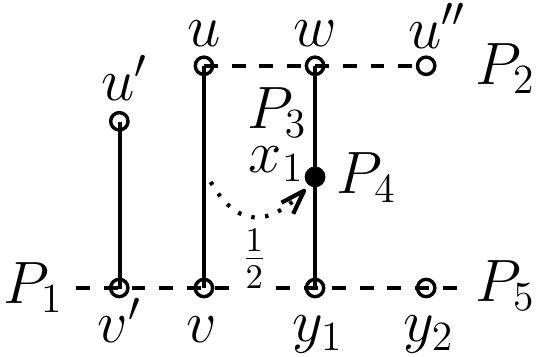}
\caption{\label{fig33b}}
\end{subfigure}
\hskip .15cm
\begin{subfigure}{0.3\textwidth}
\centering
\includegraphics[width=.85\linewidth]{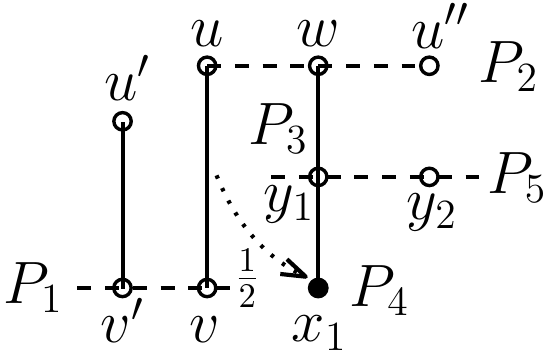}
\caption{\label{fig33c}}
\end{subfigure}
\caption{The cases when $P_4$ is a singleton of $\mcQ_1^*$, where solid edges are in $E(\mcQ_2)$ or $E(\mcQ_3)$ and dashed edges are in $E(\mcQ^*)$.
	$x_1$ is the tail of the pipe through which the destination path $P_4$ could receive $1/2$ token from the source $2$-path $u$-$v$.\label{fig33}}
\end{figure}

	{\bf Sub-case 1.1.} $w$ is the midpoint of $P_3 = x_1$-$w$-$y_1$ (see Figure~\ref{fig33a} for an illustration).
	If $P_5 \in \mcQ^*_3$, with $(y_1, y_2)$ being an edge on $P_5$,
	then $y_2$ cannot be on a $2$-path of $\mcQ_2$
	(suppose $y_2$ is on a $2$-path $P'' \in \mcQ_2$,
	then the three $2$-paths $u$-$v$, $u'$-$v'$, $P''$, and the $3$-path $P_3$ could be replaced due to {\sc Operation $3$-$1$-By-$0$-$3$}).
	Therefore, only through pipe $u'' \to w \to x_1$ could $P_4$ receive another $1/2$ token.
	Thus, $\tau_2(P_4) \le 1/2 \times 2 = 1$, implying $\tau(P_4) \le 0 + 1 = 1$.

	{\bf Sub-case 1.2.} $w$ is an endpoint of $P_3$, {\it i.e.}, either $P_3 = w$-$x_1$-$y_1$ (see Figure~\ref{fig33b} for an illustration)
	or $P_3 = w$-$y_1$-$x_1$ (see Figure~\ref{fig33c} for an illustration).
	In each sub-case, $u''$ cannot be the head of any pipe 
	({\it i.e.}, there does not exist a path $u''$-$v''$-$v'''$-$u'''$,
	where $u''$-$v''$, $v'''$-$u''' \in \mcQ_2$ and $v''$-$v''' \in \mcQ^*_2$,
	since otherwise, the four $2$-paths $u$-$v$, $u'$-$v'$, $u''$-$v''$, $v'''$-$u'''$, and the $3$-path $P_3$ could be replaced
	due to {\sc Operation $4$-$1$-By-$1$-$3$}).
	If $P_5 \in \mcQ^*_3$, with $(y_1, y_2)$ being an edge on $P_5$,
	then $y_2$ in Figure~\ref{fig33b} cannot be on a $2$-path of $\mcQ_2$
	(suppose $y_2$ is on a $2$-path $P'' \in \mcQ_2$,
	then the three $2$-paths $u$-$v$, $u'$-$v'$, $P''$, and the $3$-path $P_3$ could be replaced due to {\sc Operation $3$-$1$-By-$0$-$3$});
	$y_2$ in Figure~\ref{fig33c} cannot be the head of any pipe
	({\it i.e.}, there does not exist a path $y_2$-$z$-$z'$-$y'$,
	where $y_2$-$z$, $z'$-$y' \in \mcQ_2$ and $z$-$z' \in \mcQ^*_2$,
	since otherwise, the three $2$-paths $u$-$v$, $y_2$-$z$, $z'$-$y'$, and the $3$-path $P_3$ could be replaced due to {\sc Operation $3$-$1$-By-$0$-$3$}).
	Therefore, through no other pipe could $P_4$ receive any other token in either sub-case.
	Thus, $\tau_2(P_4) \le 1/2$, implying $\tau(P_4) \le 0 + 1/2 = 1/2$.

{\bf Case 2.}
Both $P_4$ and $P_5$ are paths of $\mcQ_2^* \cup \mcQ_3^*$ (see Figure~\ref{fig34} for illustrations).
We distinguish two sub-cases based on whether $w$ is an endpoint or the midpoint of $P_3$ to determine
how to transfer the $1/2$ token from the source $2$-path $u$-$v$ to $P_4$ or $P_5$ or both.

\begin{figure}[ht]
\centering
\captionsetup[subfigure]{justification=centering}
\begin{subfigure}{0.45\textwidth}
\centering
\includegraphics[width=.6\linewidth]{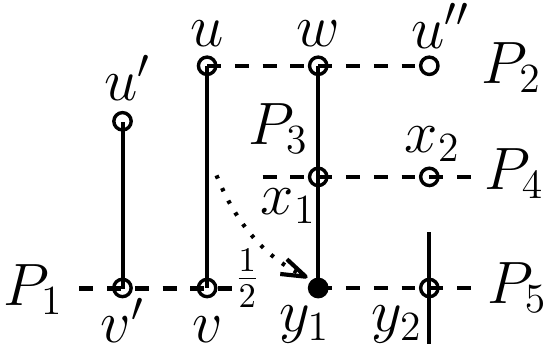}
\caption{\label{fig34a}}
\end{subfigure}
\hskip .15cm
\begin{subfigure}{0.45\textwidth}
\centering
\includegraphics[width=.58\linewidth]{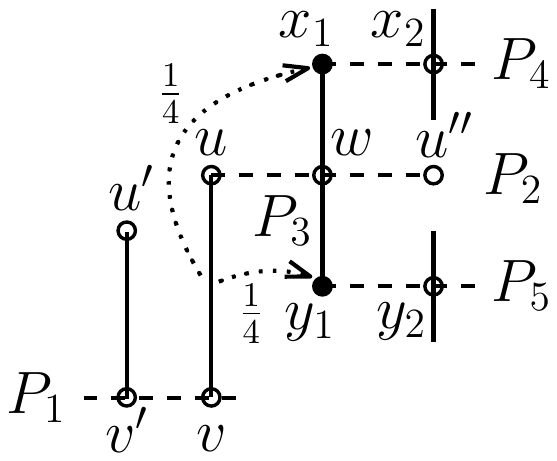}
\caption{\label{fig34b}}
\end{subfigure}
\caption{The cases when both $P_4$ and $P_5$ are in $\mcQ_2^* \cup \mcQ_3^*$, where solid edges are in $E(\mcQ_2)$ or $E(\mcQ_3)$ and dashed edges are in $E(\mcQ^*)$.
	In Figure~\ref{fig34a}, $y_1$ is the tail of the pipe through which $P_5$ receives $1/2$ token from the source $2$-path $u$-$v$;
	in Figure~\ref{fig34b}, $x_1$ is the tail of the pipe through which $P_4$ receives $1/4$ token from the source $2$-path $u$-$v$
	and $y_1$ is the tail of the pipe through which $P_5$ receives $1/4$ token from the source $2$-path $u$-$v$.\label{fig34}}
\end{figure}

	{\bf Sub-case 2.1.} $w$ is an endpoint of $P_3 = w$-$x_1$-$y_1$, with $y_1$ on $P_5$ (see Figure~\ref{fig34a} for an illustration).
	In this sub-case, we transfer the $1/2$ token from the source $2$-path $u$-$v$ to $P_5$ through pipe $u \to w \to y_1$. 
	Similar to the sub-case shown in Figure~\ref{fig33b},
	if $(y_1, y_2)$ is an edge on $P_5$, 
	then $y_2$ cannot be on a $2$-path of $\mcQ_2$ due to {\sc Operation $3$-$1$-By-$0$-$3$}.
	Thus, through no other pipe with tail $y_1$ could $P_5$ receive any other token.
	Therefore, $P_5$ could receive at most $1/2$ token through pipes with tail $y_1$.

	{\bf Sub-case 2.2.} $w$ is the midpoint of $P_3 = x_1$-$w$-$y_1$ (see Figure~\ref{fig34b}).
	In this sub-case, we break the $1/2$ token holding by the source $2$-path $u$-$v$ into two quarters,
	with $1/4$ transferred to $P_4$ through pipe $u \to w \to x_1$ and
	the other $1/4$ transferred to $P_5$ through pipe $u \to w \to y_1$. 
	Similar to the sub-case shown in Figure~\ref{fig33a}, 
	if $(x_1, x_2)$ is an edge on $P_4$ (or $(y_1, y_2)$ is an edge on $P_5$, respectively),
	then $x_2$ (or $y_2$, respectively) cannot be on a $2$-path of $\mcQ_2$ due to {\sc Operation $3$-$1$-By-$0$-$3$}.
	Thus, only through pipe $u'' \to w \to x_1$ could $P_4$ receive another $1/4$ token and
	only through pipe $u'' \to w \to y_1$ could $P_5$ receive another $1/4$ token.
	Therefore, $P_4$ ($P_5$, respectively) could receive at most $1/2$ token through pipes with tail $x_1$ ($y_1$, respectively).

Now we discuss if $P_4$ in Figure~\ref{fig34b} and $P_5$ in Figure~\ref{fig34a} and Figure~\ref{fig34b} could receive more token through pipes
with vertices other than $x_1$ and $y_1$ being the tail, respectively.
Let $(x_1, x_2)$ and $(y_1, y_2)$ be edges on $P_4$ and $P_5$, respectively.
We first prove the following two claims.

\begin{claim}
\label{clm3.5}
The vertex $y_2$ in Figure~\ref{fig34a} and in Figure~\ref{fig34b} is on a $3$-path of $\mcQ_3$; so is the vertex $x_2$ in Figure~\ref{fig34b}.
\end{claim}
\begin{proof}
Firstly, we have already proved in the discussion for Sub-cases 2.1 and 2.2 that $y_2$ in Figure~\ref{fig34a}, and in Figure~\ref{fig34b},
and $x_2$ in Figure~\ref{fig34b} cannot be on a $2$-path of $\mcQ_2$.
Suppose $x_2$ in Figure~\ref{fig34b} is a singleton of $\mcQ_1$, 
then the $3$-path $P_3$ and the edge $(x_1, x_2)$ could be reconnected into two $2$-paths, implying $\mcQ$ not a partition with the least $1$-paths, a contradiction.
This argument also applies to $y_2$ in Figure~\ref{fig34a} and Figure~\ref{fig34b}.
Thus, the claim is proved.
\end{proof}

Claim~\ref{clm3.5} implies that for $P_5$ in Figure~\ref{fig34a} or \ref{fig34b} ($P_4$ in Figure~\ref{fig34b}, respectively),
we have $\tau_1(P_5) \le 1/2$ ($\tau_1(P_4) \le 1/2$, respectively).

\begin{claim}
\label{clm3.6}
The vertex $y_2$ in Figure~\ref{fig34a} or in Figure~\ref{fig34b} cannot be the tail of a pipe; neither can the vertex $x_2$ in Figure~\ref{fig34b}.
\end{claim}
\begin{proof}
We only prove that $y_2$ in Figure~\ref{fig34a} cannot be the tail of a pipe.
The same argument applies to $y_2$ in Figure~\ref{fig34b} and $x_2$ in Figure~\ref{fig34b}.
Suppose $y_2$ in Figure~\ref{fig34a} is the tail of a pipe, say $z_1 \to w' \to y_2$.
That is, $y_2$ and $w'$ are on the same $3$-path, say $P'_3$, of $\mcQ_3$;
there is a path $w'$-$z_1$-$z_2$-$z_3$-$z_4$, 
where $w'$-$z_1$, $z_2$-$z_3 \in \mcQ^*_2$ and $z_1$-$z_2$, $z_3$-$z_4 \in \mcQ_2$.
(See Figure~\ref{fig35} for an illustration.)
Then the four $2$-paths $u$-$v$, $u'$-$v'$, $z_1$-$z_2$, $z_3$-$z_4$, and the two $3$-paths $P_3$ and $P'_3$ could be replaced due to {\sc Operation $4$-$2$-By-$1$-$4$}.
Thus, the claim is proved.
\end{proof}
\begin{figure}[ht]
\centering
\includegraphics[width=.36\linewidth]{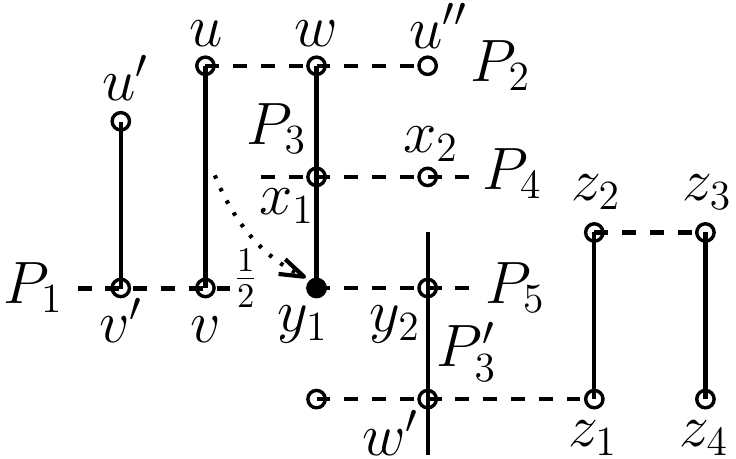}
\caption{An illustration of $y_2$ in Figure~\ref{fig34a} being the tail of a pipe, say $z_1 \to w' \to y_2$,
	which could never happen due to {\sc Operation $4$-$2$-By-$1$-$4$}.\label{fig35}}
\end{figure}

Claim~\ref{clm3.6} implies that
through no pipe with tail $y_2$ in Figure~\ref{fig34a} or in Figure~\ref{fig34b} 
(with tail $x_2$ in Figure~\ref{fig34b}, respectively) 
could $P_5$ ($P_4$, respectively) receive any other token.
That is, if $P_5$ ($P_4$, respectively) is a $2$-path or a $3$-path with $y_1$ ($x_1$, respectively) being the midpoint,
then it could receive token only through pipes with tail $y_1$ ($x_1$, respectively),
thus we have $\tau_2(P_5) \le 1/2$ ($\tau_2(P_4) \le 1/2$, respectively).

Next, we discuss the cases when $P_5$ in Figure~\ref{fig34a} is a $3$-path, with $y_1$ being an endpoint.
The same argument applies to 
the cases when $P_5$ in Figure~\ref{fig34b} is a $3$-path, with $y_1$ being an endpoint,
and the cases when $P_4$ in Figure~\ref{fig34b} is a $3$-path with $x_1$ being an endpoint.
Let $P_5 = y_1$-$y_2$-$y_3$ (see Figure~\ref{fig36} for an illustration).
According to Claim~\ref{clm3.6},
$P_5$ could only receive token through pipes with tail $y_1$ or $y_3$.
We distinguish the following three cases based on whether $y_3$ is on a path of $\mcQ_1$, or $\mcQ_2$, or $\mcQ_3$.

\begin{itemize}
\parskip=2pt
	\item If $y_3$ is a singleton of $\mcQ_1$, then we have $\tau_1(P_5) = 0$,
	and thus with the $1/2$ token received through pipe $u \to w \to y_1$,
	we have $\tau_2(P_5) \le 1/2$, implying $\tau(P_5) \le 1/2$.
	\item If $y_3$ is on a $2$-path of $\mcQ_2$, then we have $\tau_1(P_5) \le 1/2$,
	and thus with the $1/2$ token received through pipe $u \to w \to y_1$,
	we have $\tau_2(P_5) \le 1/2$, implying $\tau(P_5) \le 1$.
	\item If $y_3$ is on a $3$-path of $\mcQ_3$, then we have $\tau_1(P_5) = 0$.
	$y_3$ could either be the tail of a pipe as $y_1$ in Sub-case 2.1 (Figure~\ref{fig34a}),
	or be the tail of at most two pipes as $x_1$ or $y_1$ in Sub-case 2.2 (Figure~\ref{fig34b}).
	For any of these sub-cases,
	$P_5$ could receive at most $1/2$ token through pipes with tail $y_3$.
	Thus, with the $1/2$ token received through pipe $u \to w \to y_1$,
	we have $\tau_2(P_5) \le 1/2 + 1/2 = 1$, implying $\tau(P_5) \le 0 + 1 = 1$.
\end{itemize}

\begin{figure}[ht]
\centering
\includegraphics[width=.38\linewidth]{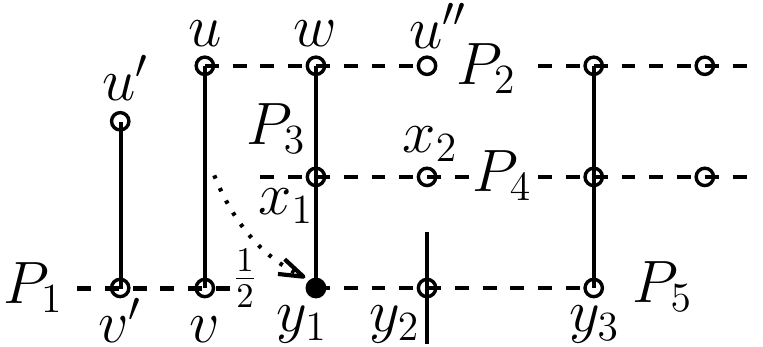}
\caption{An illustration of $y_1$ being an endpoint of $P_5$ in Figure~\ref{fig34a},
	where $P_5 = y_1$-$y_2$-$y_3$, solid edges are in $E(\mcQ)$ and dashed edges are in $E(\mcQ^*)$.
	$y_3$ could be on a path of $\mcQ_1$, $\mcQ_2$, or $\mcQ_3$.\label{fig36}}
\end{figure}

From the above discussions, we conclude that for any $P \in \{P_4, P_5\}$, 
if $\tau_2(P) > 0$, then we have $\tau_1(P) \le 1/2$ and $\tau_2(P) \le 1$, 
and it falls into one of the following four scenarios:

\begin{enumerate}
\parskip=2pt
	\item If $w$ is an endpoint of $P_3$ and $\tau_1(P) = 0$,
	then there are at most two pipes through each of which could $P$ receive $1/2$ token.
	That is, $\tau_2(P) \le 1/2 \times 2 = 1$, implying $\tau(P) \le 0 + 1 = 1$.
	\item If $w$ is an endpoint of $P_3$ and $\tau_1(P) = 1/2$,
	then only through one pipe could $P$ receive the $1/2$ token.
	That is, $\tau_2(P) \le 1/2$, implying $\tau(P) \le 1/2 + 1/2 = 1$.
	\item If $w$ is the midpoint of $P_3$ and $\tau_1(P) = 0$,
	then there are at most four pipes through each of which could $P$ receive $1/4$ token.
	That is, $\tau_2(P) \le 1/4 \times 4 = 1$, implying $\tau(P) \le 0 + 1 = 1$.
	\item If $w$ is an endpoint of $P_3$ and $\tau_1(P) = 1/2$,
	then there are at most two pipes through each of which could $P$ receive $1/4$ token.
	That is, $\tau_2(P) \le 1/4 \times 2 = 1/2$, implying $\tau(P) \le 1/2 + 1/2 = 1$.
\end{enumerate}

In summary, for any $P_1 \in \mcQ^*$ with $\tau_1(P_1) = 3/2$, 
we have $\tau_2(P_1) = -1/2$;
for any $P \in \mcQ^*$ with $\tau_2(P) > 0$,
we have $\tau_1(P) = 0$ if $\tau_2(P) \le 1$, or $\tau_1(P) \le 1/2$ if $\tau_2(P) \le 1/2$.
Therefore, at the end of Phase 2,
we have 
\begin{enumerate}
\parskip=2pt
	\item $\tau(P_i) \le 1$ for $\forall P_i \in \mcQ^*_1$,
	\item $\tau(P_j) \le 2$ for $\forall P_j \in \mcQ^*_2$,
	\item $\tau(P_\ell) \le 1$ for $\forall P_\ell \in \mcQ^*_3$.
\end{enumerate}
This proves Lemma~\ref{lm3.2}.

\subsection{A tight instance for the algorithm {\sc Approx}} \label{sec4.3}
Figure~\ref{fig37} illustrates a tight instance, in which our solution $3$-path partition $\mcQ$ contains nine $2$-paths and three $3$-paths (solid edges)
and an optimal $3$-path partition $\mcQ^*$ contains nine $3$-paths (dashed edges).
Each $3$-path of ${\cal Q}^*$ receives $1$ token from the $2$-paths in $\mcQ$ in our distribution process.
This instance shows that the performance ratio of $4/3$ is tight for {\sc Approx},
thus Theorem~\ref{thm3.1} is proved.

\begin{figure}[ht]
\centering
\includegraphics[width=0.65\linewidth]{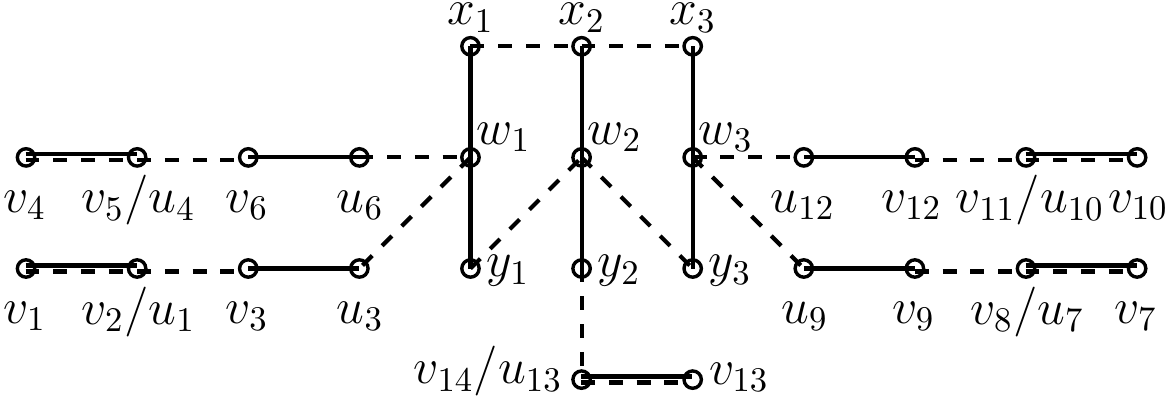}
\caption{A tight instance of $27$ vertices,
	where solid edges represent a $3$-path partition produced by {\sc Approx} and
	dashed edges represent an optimal $3$-path partition.
	The edges $(u_{3i+1}, v_{3i+1})$, $i = 0, 1, \ldots, 4$, are in $E(\mcQ_2) \cap E(\mcQ^*)$, shown in both solid and dashed.
	The vertex $u_{3i+1}$ collides into $v_{3i+2}$, $i = 0, 1, \ldots, 4$.
	In our distribution process, each of the nine $3$-paths in $\mcQ^*$ receives $1$ token from the $2$-paths in $\mcQ$. \label{fig37}}
\end{figure}

\section{Conclusions} \label{sec5}
We studied the {\sc $3$PP} problem and designed a $4/3$-approximation algorithm {\sc Approx}.
{\sc Approx} first computes a $3$-path partition $\mcQ$ with the least $1$-paths in $O(nm)$-time,
then iteratively applies four local operations with different priorities to reduce the total number of paths in $\mcQ$.
The overall running time of {\sc Approx} is $O(n^6)$.
The performance ratio $4/3$ of {\sc Approx} is proved through an amortization scheme, 
using the structure properties of the $3$-path partition returned by {\sc Approx}.
We also show that the performance ratio $4/3$ is tight for our algorithm.

The {\sc $3$PP} problem is closely related to the {\sc $3$-Set Cover} problem, but none of them is a special case of the other.
The best $4/3$-approximation for {\sc $3$-Set Cover} has stood there for more than three decades;
our algorithm {\sc Approx} for {\sc $3$PP} has the approximation ratio matches up to this best approximation ratio $4/3$.
We leave it open to better approximate {\sc $3$PP}.

\subsection*{Acknowledgement}
YC and AZ were supported by the NSFC Grants 11771114 and 11571252;
YC was also supported by the China Scholarship Council Grant 201508330054.
RG, GL and YX were supported by the NSERC Canada.
LL was supported by the China Scholarship Council Grant No. 201706315073, and the Fundamental Research Funds for the Central Universities Grant No. 20720160035.
WT was supported in part by funds from the College of Engineering and Computing at the Georgia Southern University.



\end{document}